\documentclass[10pt,conference]{IEEEtran}
\usepackage{amssymb}
\usepackage{graphicx}

\newtheorem{lemma}{Lemma}

\newtheorem{theorem}{Theorem}

\begin{document}

\title{A Coding Theoretic Approach for Evaluating Accumulate Distribution on Minimum Cut Capacity 
of Weighted Random Graphs}

\author{\authorblockN{Tadashi Wadayama and  Yuki Fujii}
\authorblockA{Department of Computer Science,  Nagoya Institute of Technology\\
Email: wadayama@nitech.ac.jp\\
}
} 

\maketitle
\begin{abstract}
The multicast capacity of a directed network is closely related to the $s$-$t$ maximum flow, which 
is equal to the $s$-$t$ minimum cut capacity due to the max-flow min-cut theorem.
If the topology of a network 
(or link capacities) is dynamically changing or have stochastic nature, it is not so trivial to predict statistical properties on 
the maximum flow.
In this paper, we present a coding theoretic approach 
for evaluating the accumulate distribution of the minimum cut capacity of weighted random graphs.
The main feature of our approach is to utilize the correspondence between the cut space of a graph and 
a binary LDGM (low-density generator-matrix) code with column weight 2.
The graph ensemble treated in the paper is a weighted version of Erd\H{o}s-R\'{e}nyi random graph ensemble.
The main contribution of our work is a combinatorial lower bound for the accumulate distribution of the minimum cut capacity.
From some computer experiments, it is observed that the lower bound derived here 
reflects the actual statistical behavior of the minimum cut capacity.
\end{abstract}

\section{Introduction}

Rapid growth of information flow over a network
such as a backbone network for mobile terminals  requires efficient utilization of full potential of the network. 
In a multicast communication scenario,  it is well known that 
an appropriate network coding achieves its multicast capacity.
Emergence of the network coding have broaden 
network design strategies for 
efficient use of wired and wireless networks \cite{Ahlswede}.

The multicast capacity of a directed graph is closely related to the $s$-$t$ maximum flow, which 
is equal to the $s$-$t$ minimum cut capacity due to the max-flow min-cut theorem \cite{Schrijver}.
The topology of a network and the assignment of the link capacities determine
the $s$-$t$ minimum cut capacity of a network. 

If a network is fixed, the corresponding $s$-$t$ maximum flow of 
the network can be efficiently evaluated by using Ford-Fulkerson algorithm \cite{Schrijver}.  
However, if the topology of a network 
(or link capacities) is dynamically changing or have stochastic nature, it is not so trivial to predict statistical properties on 
the maximum flow. For example, in a case of wireless network, the link capacities may fluctuate 
because of the effect of time-varying fading. Another example is an ad-hoc network
whose link connections are stochastically determined.

In order to obtain an insight for the statistical property of the min-cut capacity for such random networks, 
it is natural to investigate the statistical properties of a random graph ensemble.  
Such a result may unveil  typical behaviors of minimum cut capacity (or maximum flow) for 
given parameters, such as the number of vertices and edges. 

Several theoretical works on the maximum flow of  random graphs ({\em i.e.}, graph ensembles) have been made.
In a context of randomized algorithms, Karger showed a sharp concentration result for maximum flow in 
the asymptotic regime \cite{Karger}.
Ramamoorthy et al.  presented another concentration result;
the network coding capacities of weighted random graphs and weighted random geometric graphs
concentrate around the expected number of nearest neighbors of the source and the sinks \cite{Aditya}.
These concentration results indicate an asymptotic property of the maximum flow of random networks.
Wang et al.  shows statistical property of the maximum flow in an asymptotic setting as well.
They discussed the random graph with  Bernoulli distributed weights \cite{Wan}.

In this paper, we present a coding theoretic approach 
for evaluating the accumulate distribution of the minimum-cut capacity of weighted random graphs.
This approach is totally different from those used in the conventional works. 
The basis of the analysis is  the correspondence between the cut space of an undirected graph \cite{CUT} and 
a binary LDGM (low-density generator-matrix) code with column weight 2.
Yano and Wadayama presented that an ensemble analysis for a class of binary LDGM codes 
with column weight 2 for the network reliability problem \cite{Yano}.
This paper extends the idea in \cite{Yano} to weighted graph ensembles.
We focus on a weighted version of Erd\H{o}s-R\'{e}nyi random graph ensemble \cite{RANDOMGRAPH} in this paper.

\section{Preliminaries}

In this section, we first introduce several basic definitions and notation used throughout the paper. 
Then, the cut-set weight distribution will  be discussed. 

\subsection{Notation and definitions}

A graph $G=(V,E)$ is a pair of a vertex set $V=\{v_1, \ldots, v_k\}$ and an edge set 
$E=\{e_1,\ldots, e_n \}$ where $e_j = (u,v), u,v \in V$ is an edge.
If $e_j = (u,v)$ is not an ordered pair, i.e., $(u,v)=(v,u)$, the graph $G$ is called {\em an undirected graph}.
Otherwise, i.e., $(u,v)$ is an ordered pair, $G$ is {\em a directed graph}.
The two vertices connecting an edge $e \in E$ are referred to as the end points of $e$.
If an edge $e=(u,u)$ has the identical end points, $e$ is called {a self-loop}.

If real valued function $w: E \rightarrow \Bbb R_{\ge 0}$ is defined for an undirected graph $G=(V,E)$,
the triple $(V, E, w)$ is considered as a {\em weighted graph}. 
The set $\Bbb R_{\ge 0}$ represents the set of non-negative real numbers.
In our context, the weight function  $w$ represents the link capacity for each edge.

Assume that a weighted undirected graph $G=(V, E, w)$ is given.
A non-overlapping bi-partition $V=X \cup (V \backslash X)$ is called a {\em cut} where 
$X$ is a non-empty proper subset of $V (X \ne V)$.
The set of edges bridging $X$ and $V \backslash X$ is referred to as the {\em cut-set} corresponding to 
the cut of $(X, V \backslash X)$. 
The cut weight of $X$ is defined as 
$
\sum_{u \in X, v \in V \backslash X} w(u,v).
$

\subsection{Random graph ensemble}
\label{randomensemble}
In the following, we will define an ensemble of weighted undirected graphs.
The graph ensemble is based on  Erd\H{o}s-R\'{e}nyi random graph ensemble.
Let $k (k \ge 1)$ be the number of labeled vertices and 
$n (1 \le n \le k(k-1)/2)$ be the number of labeled undirected edges.
The vertices are labeled from $1$ to $k$ and the edges are labeled from $1$ to $n$.

For any adjacent vertices, a single edge is only allowed.
It is assumed that each edge has own integer weight; namely, a weight $w_i \in [1,q] (i \in [1,n])$ is assigned to
the edge with label $i$, which is denoted by the $i$th edge. The notation $[a,b]$ denotes 
the set of consecutive integers from $a$ to $b$.
The set $R_{k,n}^q$ denotes the set of all undirected weighted graphs with 
$k$-vertices and $n$-edges satisfying the above assumption.

For any $G \in R_{k,n}^q$, the sets of vertices and edges are denoted by $V(G)$ and $E(G)$, respectively.
In a similar way, $w_i(G)$ is defined as the weight of $i$th edge of $G$.

It is evident that the cardinality of  $R_{k,n}^q$ is given by 
\begin{equation}
|R_{k,n}^q | = n! {{k \choose 2}\choose  n}q^n.
\end{equation}
We here assign the probability 
\begin{equation}
P(G)= \frac{1}{n! {{k \choose 2}\choose  n}} \mu(w_1(G)) \mu(w_2(G)) \cdots \mu(w_n(G))   
\end{equation}
for $G \in R^q_{k,n}$ where $\mu$ is a discrete probability measure defined over $[1,q]$; namely,
it satisfies 
\begin{equation}
\sum_{w\in [1,q] }\mu(w) = 1 \mbox{ and } \forall w \in [1,q], \mu(w) \ge 0.
\end{equation}

The pair $(R_{k,n}^q, P)$ defines an ensemble of random graphs and it is  denoted by ${\cal E}$.
\subsection{Incidence matrix}

For $G \in R^q_{k,n}$, the incidence matrix of $G$, denoted by $M(G) \in \{0,1\}^{k \times n}$,
is defined as follows:
\begin{equation}
M(G)_{i,j} = 
\left\{
\begin{array}{ll}
1, & \mbox{if $i$-th vertex connects to $j$-th edge}\\
0, & \mbox{otherwise},\\
\end{array}
\right.
\end{equation}
where $M(G)_{i,j}$ is the $(i,j)$-element of $M(G)$.
If $G$ is connected, then the rank  (over $\Bbb F_2$) of $M(G)$ is $k-1$.
The row space  (over $\Bbb F_2$) of $M(G)$ coincides with the set of all possible incidence vectors of 
cut-sets of $G$. 

\subsection{Cut weight distribution}

For a given undirected graph, we can enumerate the number of cut-sets with cut weight $w$.
The {\em cut weight distribution} of $G$ by
\begin{equation}
B_w(G) = \sum_{E \subset E(G)} \Bbb I [E \mbox{ is a cut-set of  }  G,   \mbox{cut weight is } w]
\end{equation}
for positive integer $w$. The function $\Bbb I[\cdot]$ is the indicator function that takes value 1 if the 
condition is true; otherwise it takes value 0.
This cut weight distribution can be regarded as an analog of the weight distribution of the binary linear code defined 
by the incidence matrix of a given undirected graph.

For ensemble analysis, it is convenient to introduce another form of the weight distribution.
The {\em detailed cut weight distribution}
$
A_{u,v.w}: R_{k,n}^q \rightarrow \Bbb Z_{\ge 0} 
$
is defined by
\begin{eqnarray} \nonumber
A_{u,v,w}(G)  \\
&&\hspace{-2.3cm}= \sum_{m \in Z^{(k,u)} } \sum_{c \in Z^{(n,v)} }\Bbb I \left[m M(G) = c,\sum_{i=1}^n c_i w_i(G)=w \right],
\end{eqnarray}
for $u \in [1,k-1], v \in [0,n]$.
The set of constant weight binary vectors 
$Z^{(a,b)}$ is defined as
\begin{equation}
Z^{(a,b)} = \{x \in  \{0,1\}^a: w_H(x)= b  \}.
\end{equation}
The function $w_H(\cdot)$ represents the Hamming weight.  
Assume that the cardinality of the cut is $v$ and that the size of $X \subset E$ is $u$.
Under this condition, the function $A_{u,v,w}(G)$ represents the number of cuts with the cut weight $w$.
It should be noted that the one-to-one correspondence between the cut space and the set of incident vectors of the cuts 
are implicitly used in the definition of $A_{u,v,w}(G)$.

The following lemma indicates the relationship between $A_{u,v,w}$ and $B_w(G)$.
\begin{lemma}\label{bwg}
For $G \in R_{k,n}^q$,  the cut weight distribution $B_w(G)$ can be upper bounded by
\begin{equation}\label{upperbound}
B_w(G) \le \frac 1 2 \sum_{u=1}^{k-1} \sum_{v=0}^n A_{u,v.w}(G),
\end{equation}
for  $w \in \Bbb Z_{\ge 0}$. The notation $\Bbb Z_{\ge 0}$ represents the set of non-negative integers.
\end{lemma}
\begin{proof}
Let 
\[
S(G) = \{c = m M(G)  \in \Bbb F_2^n  \mid m \in \Bbb F_2^k \backslash \{0^k, 1^k\} \}.
\]
The cut weight distribution $B_w(G)$ can be rewritten as follows.
\begin{eqnarray} \nonumber
B_w(G) 
&=& \sum_{E \subset E(G)} \Bbb I [E \mbox{ is a cut-set of  }  G,   \mbox{cut capacity  is } w] \\ \label{abcd}
&=& \sum_{c \in  S(G) }  \Bbb I \left[\sum_{i=1}^n c_i w_i(G)=w \right]. 
\end{eqnarray}
The second equality is due to the fact that the row space of $M(G)$ equals the set of all possible 
cut-set vectors of $G$.
It is evident that 
\[
|\{m \in \Bbb F_2^k \mid m \ne 0^k, m \ne 1^k, c = m M(G)  \}  |  \ge 2
\] 
holds for any $c \in S(G)$. This implies that 
\begin{equation} \label{rowspace}
\sum_{c \in S(G) } h(c) \le  \frac {1}{2} \sum_{m \in \{0,1\}^k \backslash \{0^k, 1^k\}} \sum_{c \in \{0,1\}^n} h(c)\Bbb I \left[c = m M(G) \right]
\end{equation}
holds for any real-valued function $h: \{0,1\}^n \rightarrow \Bbb R$.
Substituting (\ref{rowspace}) into (\ref{abcd}), we obtain 
\begin{eqnarray} \nonumber
B_w(G)
 &\le& \frac 1 2 \sum_{m \in \{0,1\}^k \backslash \{0^k, 1^k\}} 
\sum_{c \in \{0,1\}^n} \Bbb I \left[c = m M(G) \right]  \\ \nonumber
&\times & \Bbb I \left[  \sum_{i=1}^n c_i w_i(G)=w \right] \\
&=& \frac 1 2 \sum_{u=1}^{k-1} \sum_{v=0}^n A_{u,v.w}(G).
\end{eqnarray}
\end{proof}

\section{Ensemble average of cut weight distribution}

In this section, we discuss the  average of $A_{u,v,w}(G)$ over the ensemble ${\cal E}$.
This analysis is very similar to the derivation of the 
average weight distribution of LDGM codes with column weight 2.

\subsection{Preparation}

In the following, the expectation operator ${\sf E}$ is defined as
\begin{equation}
{\sf E}[f(G)] = \sum_{G \in R_{k,n}^q } P(G) f(G),
\end{equation}
where $f$ is any real-valued function defined on $R_{k,n}^q$.
The next lemma plays a key role to derive a closed form the average cut set weight distribution.
\begin{lemma} \label{closedlemma}
Assume that  $m^* \in \{0,1\}^k $ and 
$c^* \in \{0,1\}^n $ satisfies $w_H(m^*)=u$ and $w_H(c^*)=v$
where $u \in [1,k-1]$ and $v \in [0,n]$. The following equality 
\begin{eqnarray}  \nonumber
&&\hspace{-1cm}{\sf E} \left[ \Bbb I \left[m^* M(G) = c^*, \sum_{i=1}^n c_i^* w_i(G)=w \right]  \right] \\  \nonumber
&=& \frac{1}{{n \choose v} {{k \choose 2}\choose  n}}{u(k-u) \choose v} {{k \choose 2} - u(k-u) \choose n - v}
[x^w] f(x)^v \\ \label{closed}
\end{eqnarray}
holds. The function $f(x)$ is defined by
\begin{equation}
f(x) = \sum_{i=1}^q \mu(i) x^{i}.
\end{equation}
The term $[x^w] f(x)^v$ represents the coefficient of $x^w$ in $f(x)^v$.
\end{lemma}

\begin{proof}
Due to the symmetry of the ensemble,  we can assume that
the first $u$-elements of $m^*$ are one and the rests are zero without loss of generality.
In a similar manner, $c^*$ is assumed to be the binary vector such that
first $v$-elements are one and the rests are zero.

In the following, we will count the number of labeled graphs satisfying $m^* M(G) = c^*$ by
counting the number of binary incidence matrices satisfying the above condition.
Let 
$
M(G)=  \left( f_1 \  f_2 \ \cdots f_n \right)
$
where $f_i$ is the $i$th column vector of $M(G)$.
Since $M(G)$ is an incidence matrix, the column weight of $f_i$ is $w_H(f_i) = 2$ for $i \in [1,n]$.
From the assumptions described  above,  we have
\begin{equation} \label{mfi}
m^* f_i = 
\left\{
\begin{array}{ll}
1, & i \in [1,v] \\
0, & i \in [v+1,n]. \\
\end{array}
\right.
\end{equation}
We then count the number of allowable combinations of $(f_1,f_2,\ldots, f_n)$ satisfying (\ref{mfi}).
Let 
\begin{equation}
A = \{f \in \{0,1\}^k \mid  m^* f = 1, w_H(f)=2\}.
\end{equation}
The cardinality of $A$ is given by $|A |= u(k-u)$ because 
a non-zero component of $f$ needs to have an index within $[1,u]$ and
another non-zero component has an index in the range $[u+1, k]$.
This observation leads to the number of possibilities for $(f_1,f_2,\ldots, f_v)$
which is given by
$
v! {u(k-u) \choose v}.
$
The remaining $n-v$ columns, $(f_{v+1}, \ldots, f_n)$, should be taken from 
the set $\{f \in \{0,1\}^k \mid  w_H(f)=2\} \backslash A$.
Thus, the number of possibilities for such choice is 
$
(n-v)! {{k \choose 2} - u(k-u) \choose n - v}.
$
In summary, the number of allowable combinations of $(f_1,f_2,\ldots, f_n)$ denoted by $S$
is given by
\begin{equation}\label{cardx}
S = v! (n-v)! {u(k-u) \choose v}{{k \choose 2} - u(k-u) \choose n - v}.
\end{equation}

We are now ready to derive the claim of this lemma. 
To simplify the notation, the cut weight is denoted by
$
\psi(c) = \sum_{i=1}^n c_i w_i(G).
$
The left hand side of (\ref{closed}) can be rewritten as follows:
\begin{eqnarray} \nonumber
&& \hspace{-1cm} {\sf E}[ \Bbb I[m^* M(G) = c^*, \psi(c^*)=w]   \\ \nonumber
&=& \sum_{G \in R_{k,n}^q}  P(G) I[m^* M(G) = c^*, \psi(c^*)=w] \\ \nonumber
&=&\frac{1}{n! {{k \choose 2}\choose  n}} \sum_{G \in R_{k,n}^q} \left(\prod_{i\in [1,n]}   \mu(w_i(G)) \right)   \\ \nonumber
&\times &  \Bbb I[m^* M(G) = c^*, \psi(c^*)=w] \\ \nonumber 
&=&\frac{S}{n! {{k \choose 2}\choose  n}}    \sum_{\stackrel{u_1+\cdots +u_q = v,}{u_1+2 u_2 + \cdots q u_q = w}}
\left(\prod_{i\in [1,q]}   \mu(i)^{u_i} \right)   \frac{v!}{u_1! u_2 !  \cdots u_q !}\\ \nonumber
&=& \frac{v! (n-v)! }{n! {{k \choose 2}\choose  n}}{u(k-u) \choose v} {{k \choose 2} - u(k-u) \choose n - v}
[x^w] f(x)^v. \\
\end{eqnarray}
The last equality is due to  (\ref{cardx}). 
\end{proof}

The following lemma provides the ensemble average ${\sf E}[A_{u,v,w}(G)]$, 
which is a natural consequence of Lemma \ref{closedlemma}.
\begin{lemma} \label{auvw}
The expectation of $A_{u,v,w}(G)$ is given by
\begin{eqnarray} \nonumber
{\sf E}[A_{u,v,w}(G)] \hspace{6cm} \\
= \frac{1}{{{k \choose 2}\choose  n}}
{k \choose u} {u(k-u) \choose v} {{k \choose 2} - u(k-u) \choose n - v} [x^w] f(x)^v,
\end{eqnarray}
where $u \in [1,k-1], v \in [0,n], w \in \Bbb Z_{>0}$. 
\end{lemma}
\begin{proof}
The expectation of $A_{u,v,w}(G)$ can be simplified as follows:
\begin{eqnarray} \nonumber
&& \hspace{-1.3cm} {\sf E}[A_{u,v,w}(G)]  \\ \nonumber
&=& {\sf E}  \left[ \sum_{m \in Z^{(k,u)} } \sum_{c \in Z^{(n,v)} }\Bbb I[m M(G) = c,  \psi(c)=w ]   \right] \\  \nonumber
&=&  \sum_{m \in Z^{(k,u)} } \sum_{c \in Z^{(n,v)} } {\sf E}  \left[\Bbb I[m M(G) = c, \psi(c)=w] \right]\\ \label{aaa}
&=&
{k \choose u} {n \choose v} {\sf E}  \left[ \Bbb I[m^* M(G) = c^*, \psi(c^*)=w] \right].
\end{eqnarray}
The last equality is due to the symmetry of the ensemble. The binary vectors $m^* \in \{0,1\}^k $ and 
$c^* \in \{0,1\}^n $ are arbitrary vectors satisfying $w_H(m^*)=u$ and $w_H(c^*)=v$.
Substituting (\ref{closed}) in the previous Lemma  into (\ref{aaa}), we obtain the claim of this lemma.
\end{proof}

\subsection{Upper bound on average cut weight distribution}

In order to investigate statistical properties of the minimum cut weight, 
it is natural to study the tail of the average cut weight distribution. 
The following theorem provides an upper bound on average cut weight distribution 
that is  the  basis of  our analysis.
\begin{theorem}
The expectation of $B_w(G)$ over ${\cal E}$ can be upper bounded by
\begin{eqnarray} \nonumber
{\sf E}[B_w(G)] 
&\le&  \frac{1}{ 2{{k \choose 2}\choose  n}} \sum_{u=1}^{k-1} {k \choose u} \sum_{v=0}^{n}
 {u(k-u) \choose v}   \\
&\times& {{k \choose 2} - u(k-u) \choose n - v} [x^w] f(x)^v 
\end{eqnarray}
for  $w \in \Bbb Z_{\ge 0}$.
\end{theorem}

\begin{proof}
Due to Lemma \ref{bwg}, we immediately have 
\begin{eqnarray}
{\sf E}[B_w(G)]  
&\le& {\sf E} \left[ \frac 1 2 \sum_{u=1}^{k-1} \sum_{v=0}^n A_{u,v.w}(G) \right] \\
&=&  \frac 1 2 \sum_{u=1}^{k-1} \sum_{v=0}^n {\sf E} \left[A_{u,v.w}(G) \right].
\end{eqnarray}
Substituting the left hand side of the equality in Lemma \ref{auvw} into the last equation gives the claim of the theorem.
\end{proof}

\subsection{Accumulate cut distribution}
Let us define the accumulate cut weight of $G$, $C_{\delta}(G)$, by 
\begin{equation}
C_{\delta}(G) = \sum_{w=0}^{\delta-1} B_w (G),
\end{equation}
where $\delta$ is a non-negative integer.
If $C_{\delta}(G)$ is zero, the graph $G$ does not contain a cut with weight smaller than $\delta$.
This implies that $\lambda(G) \ge \delta$ in such a case, thus we have 
\begin{eqnarray} \nonumber
Pr[\lambda(G) \ge \delta]  &=& Pr[C_\delta(G) = 0] \\
&=& 1- Pr[C_\delta(G) \ge 1]. 
\end{eqnarray}
The second equality is due to the non-negativity of $C_\delta(G)$.
The probability $Pr[\lambda(G) \ge \delta]$ can be considered as the 
accumulate probability distribution for the minimum cut capacity: 
\begin{equation}
Pr[\lambda(G) \ge \delta] = \sum_{G \in R_{k,n}^q } P(G) \Bbb I[\lambda(G) \ge \delta ].
\end{equation}

The following theorem is the main contribution of this work. 
\begin{theorem}\label{maintheorem}
Assume that an ensemble ${\cal E}$ is given. The probability $Pr[\lambda(G) \ge \delta]$
can be lower bounded by 
\begin{eqnarray} \nonumber 
Pr[\lambda(G) \ge \delta]  &\ge & 1-   \frac{1}{ 2{{k \choose 2}\choose  n}}
\sum_{w=0}^{\delta-1} \sum_{u=1}^{k-1} {k \choose u} \sum_{v=0}^{n}
 {u(k-u) \choose v}   \\
&\times& {{k \choose 2} - u(k-u) \choose n - v} [x^w] f(x)^v 
\end{eqnarray}
for $\delta \in \Bbb Z_{\ge 0}$.
\end{theorem}
\begin{proof}
Markov inequality 
\begin{equation}
Pr[C_{\delta}(G) \ge 1] \le {\sf E}[C_\delta(G)]
\end{equation}
provides an lower bound on $Pr[\lambda(G) \ge \delta]$: 
\begin{eqnarray} \nonumber 
Pr[\lambda(G) \ge \delta]  &=& 1- Pr[C_\delta(G) \ge 1] \\
&\ge & 1-  {\sf E}[C_\delta(G)]. 
\end{eqnarray}
\end{proof}

\section{Numerical result}

In order to evaluate the tightness of the lower bound shown in Theorem \ref{maintheorem}, 
we made the following computer experiments. In an experiment, we generated $10^4$-instances of undirected graphs
from the random graph ensemble defined in the Sec \ref{randomensemble}.  We assumed that
$q=5$ and $\mu(1)=0.1, \mu(2)=0.2,  \mu(3)=0.4, \mu(4)=0.2,  \mu(5)=0.1$; namely, 
\begin{equation}
f(x) = 0.1 x^1 +  0.2 x^2 +  0.4 x^3 +  0.2 x^4 +  0.1 x^5.
\end{equation}
The minimum cut capacity for 
each instance was computed by using the Ford-Fulkerson algorithm \cite{Schrijver}.

Figure \ref{fig02} presents the accumulate distribution of minimum cut capacity when
the number of vertices  and edges are $k=100$ and $n=400$, respectively.
The lower curve represents the lower bound presented in Theorem \ref{maintheorem} and
the upper curve is approximate values $Pr[\lambda(G) \ge \delta]$ obtained from $10^4$-randomly 
generated instances.
We can observe that two curves shows reasonable agreement in the range $1 \le \delta \le 4$. 

Figure \ref{fig03} deals with a denser graph ensemble compared with that used in Fig. \ref{fig02}.
In this case, two curves are very close in the range $1 \le \delta \le 15$. 
Compared with Fig. \ref{fig02}, we can see that the proposed lower bound becomes tighter 
for a denser graph ensemble.

From these these experimental results, it can be said that 
the proposed lower bound captures the accumulate distribution $Pr[\lambda(G) \ge \delta]$ of the min-cut capacity of the random graph ensemble fairly well.

\begin{figure}[htbp]
\begin{center}
\includegraphics[width=1.0 \linewidth]{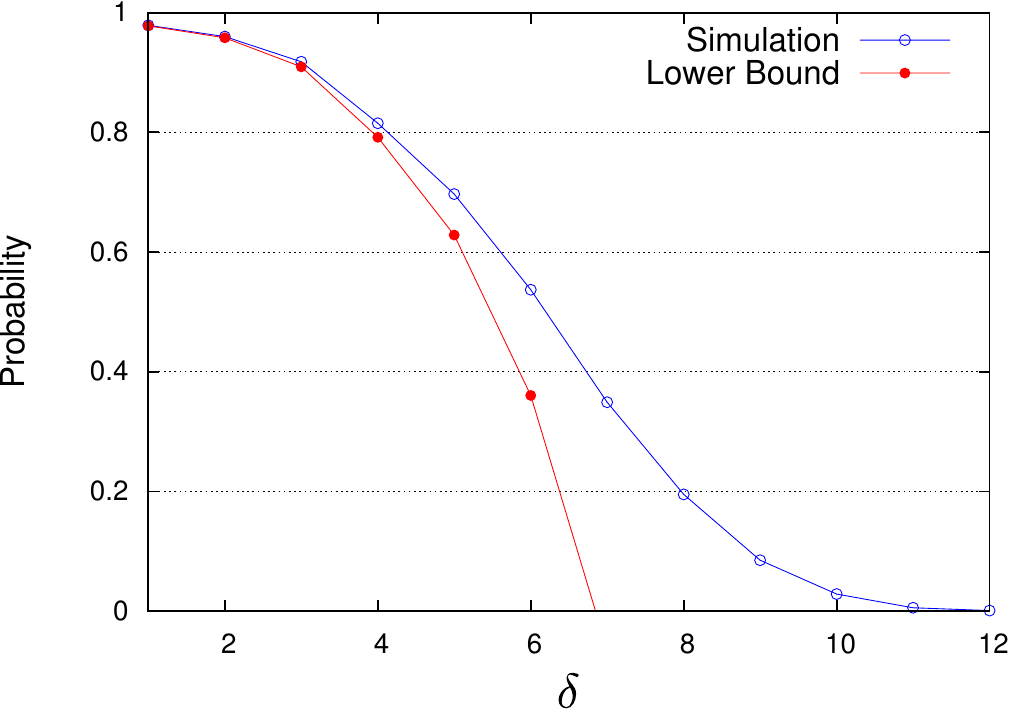} 
\end{center}
	\caption{Sparse case: accumulate distribution  of the minimum cut capacity $Pr[\lambda(G) \ge \delta]$ $(k=100, n=400)$: estimation by simulations and lower bounds}
	\label{fig02}
\end{figure}

\begin{figure}[htbp]
\begin{center}
\includegraphics[width=1.0 \linewidth]{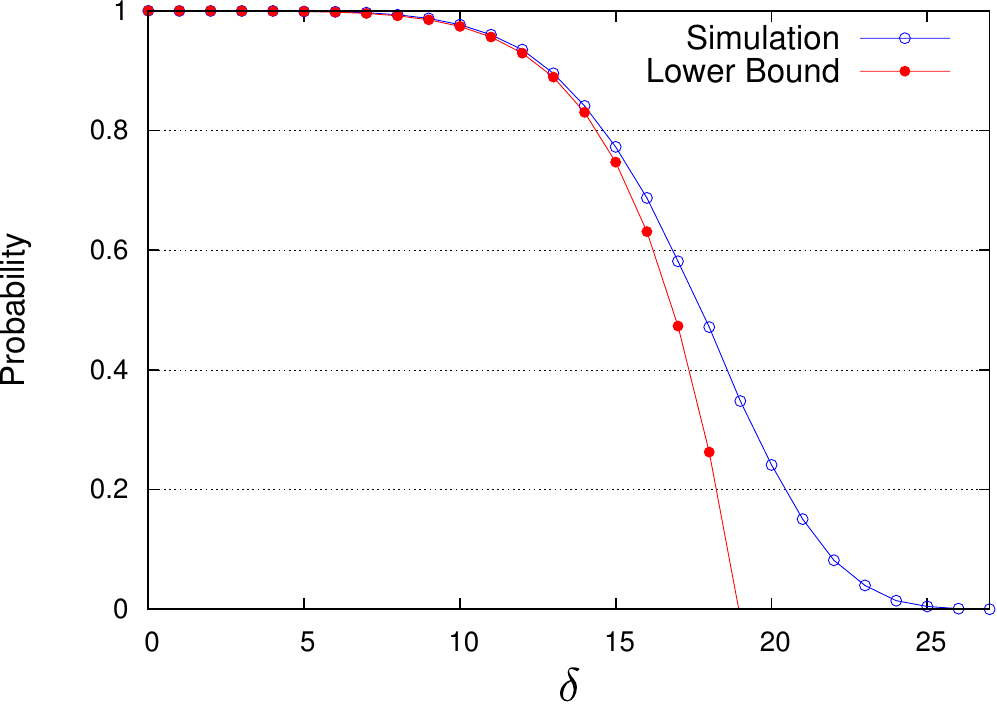} 
\end{center}
	\caption{Dense case: accumulate distribution of the minimum cut capacity $Pr[\lambda(G) \ge \delta]$ $(k=100, n=700)$ estimation by simulations and lower bounds}
	\label{fig03}
\end{figure}

\section{Conclusion}

In this paper, 
a lower bound on the accumulate distribution of the minimum cut capacity for a random graph ensemble is presented.
From the compute experiments, it is observed that the lower bound reflects actual statistical behavior of the minimum cut capacity.
The bound and the proof technique presented in the paper would
deepen our understanding on typical behaviors of the minimum-cut capacity.

The proof technique used here has close relationship to the analysis for the average weight distributions of
LDGM codes with column weight 2. The one-to-one correspondence between a cut space and the row space of 
an incidence matrix implies that the minimum cut capacity is an analog of the minimum distance of the binary linear code 
defined by an incidence matrix.  The analysis presented here  has similarity to the typical minimum distance analysis of
LDPC code ensembles \cite{Gallager}. 

An advantage of the proposed technique is its applicability for a graph ensemble with finite number of vertices and edges.
Most related studies deal with asymptotic behaviors and cannot directly be applied to a finite size graph ensemble.
Of course, it would be interesting to investigate the asymptotic behavior of the proposed lower bound when $n$ and $k$ approach 
infinity while maintaining the relationship $n = f(k)$ ($f$ is a real-valued function, e.g., $n = \beta k^2$). 

The second advantage of the proposed technique is extensibility.
In this paper, we discussed a simple graph ensemble, which is closely related to 
the Erd\H{o}s-R\'{e}nyi random graph ensemble \cite{RANDOMGRAPH}. 
The analysis for deriving the average cut-set weight distribution is approximately equivalent 
to the analysis of the average weight distribution of an LDGM code ensemble \cite{Hu} 
or of the average coset weight distribution of an LDPC code ensemble \cite{wadayama2}. 
Extension to other graph ensembles, such as regular or irregular bipartite graph ensembles, 
 may be straightforward.

\subsection*{Acknowledgement}
This work was partly supported by the Ministry of Education, Science, Sports and Culture, Japan, Grant-in-Aid, No. 22560370.


\begin{thebibliography}{20}

\bibitem{prob}
N. Alon and J.H. Spencer, ``The probabilistic method,'' Wiley InterScience (2000).

\bibitem{Bollobas}
B. Bollobas, ``Random graph, '' (2nd ed.) Cambridge University Press, 2001.

\bibitem{Diestel}
R. Diestel, ``Graph theory, '' Springer-Verlag, New York, 2000.


\bibitem{RANDOMGRAPH} P. Erd\H{o}s and A. R\'{e}nyi, ``On random graphs I,"  Publicationes Mathematicae, 6,  pp.290-297, 1959. 

\bibitem{CUT} S. L. Hakimi and H. Frank,``Cut-set matrices and linear codes," IEEE Trans. Inform.Theory, vol.IT-11, pp.457-458, July 1965.

\bibitem{Hu}
C. H. Hsu and  A. Anastasopoulos,
``Capacity-achieving codes with bounded graphical complexity and maximum likelihood decoding,''
IEEE Trans. Inform. Theory, pp.992-1006, vol.56, no. 3, Mar.  2010.

\bibitem{LS02}
  S.Litsyn and V. Shevelev,
  ``On ensembles of low-density parity-check codes: asymptotic distance distributions,''
  {\it IEEE Trans. Inform. Theory}, vol.48,  pp.887--908, Apr. 2002.

\bibitem{wadayama2}
T. Wadayama, ``Average coset weight distribution of combined LDPC matrix ensembles,''  IEEE Trans. Inform. Theory, 
pp.4856- 4866, vol.52, no.11, Nov 2006.

\bibitem{wadayama}
T. Wadayama, ``On undetected error probability of binary matrix ensembles,''
 IEEE Trans. Inform. Theory, pp.2168-2176, vol.56, no. 5, May 2010.

\bibitem{Wan}
H. Wang, P. Fan, K. B. Letaief,
``Maximum flow and network capacity of network coding for ad-hoc networks, ''
IEEE Trans. Wireless Comm., pp. 4193--4198,  vol. 6,  no.12, Dec. 2007.

\bibitem{Karger}
D. R. Karger,
``Random sampling in cut, flow, and network design problems, ''
Mathematics of Operations Research, vol. 24, no.2, pp.383--413, May 1999.

\bibitem{Aditya}
A. Ramamoorthy, J. Shi, R. D. Wesel, 
``On the Capacity of network coding for random networks,''
IEEE Trans. Inform. Theory, pp. 2878--2885, vol. 51, no.8, Aug. 2005.

\bibitem{Yano}
A. Yano and T. Wadayama, 
``Probabilistic analysis of the network reliability problem on a random graph ensemble, ''
arXiv:1105.5903, 2011.

\bibitem{modern}
T. Richardson and R. Urbanke, ``Modern coding theory,''  Cambridge University Press, 2008.

\bibitem{Schrijver}
A. Schrijver,
``Combinatorial optimization, polyhedra and efficiency,''
Springer-Verlag Berlin, 2003.

\bibitem{Ahlswede}
R. Ahlswede,  N.Cai, S.Li, and R.Yeung, ``Network information flow, " IEEE Trans. on Inform. Theory, vol.46, pp.1204--1216, Apr.  2000

\bibitem{Gallager}
R.G.Gallager, “Low density parity check codes, ” MIT Press 1963.

\end{thebibliography}
\end{document}